\documentclass[sn-mathphys,iicol]{sn-jnl}

\jyear{2022}

\usepackage{bbm}
\usepackage{amsmath,amsfonts,amssymb}
\usepackage{xspace}

\usepackage[utf8]{inputenc} 
\usepackage[T2A,T1]{fontenc}
\usepackage[russian,english]{babel}

\usepackage{csquotes}

\usepackage{graphicx} 

\usepackage[algo2e,algoruled]{algorithm2e}

\theoremstyle{thmstyleone}%
\newtheorem{theorem}{Theorem}
\newtheorem{proposition}[theorem]{Proposition}%
\newtheorem{Proposition}[theorem]{Proposition}%
\newtheorem{corollary}[theorem]{Corollary}%

\newtheorem{remark}[theorem]{Remark}%

\newtheorem{definition}[theorem]{Definition}%

\newcommand{\bb}[1]{\mathbb{#1}}

\newcommand{\s}{\lambda}
\newcommand{\ELIMINE}[1]{}

\newcommand{\gradient}{\overrightarrow{\mathrm{grad}}\xspace}
\newcommand{\dif}{\overrightarrow{\mathrm{diff}}\xspace}

\begin{document}
\title[Discrete Morse Functions and Watersheds]{Discrete Morse Functions and Watersheds}

\author[1]{\fnm{Gilles} \sur{Bertrand}}
\author[2]{Nicolas Boutry} 
\author[1]{Laurent Najman} 

\email{gilles.bertrand@esiee.fr}

\email{nicolas.boutry@lrde.epita.fr}

\email{laurent.najman@esiee.fr}

\affil[1]{\orgname{Univ Gustave Eiffel, CNRS}, \orgdiv{LIGM}, \orgaddress{\postcode{F-77454} \city{Marne-la-Vall\'ee},  \country{France}}}

\affil[2]{\orgname{EPITA}, \orgdiv{Research and Development Laboratory (LRDE)}, \orgaddress{\country{France}}}

\abstract{Any watershed, when defined on a stack on a normal pseudomanifold of dimension $d$, is a pure $(d-1)$-subcomplex that satisfies a drop-of-water principle. In this paper, we introduce Morse stacks, a class of functions that are equivalent to discrete Morse functions. We show that the watershed of a Morse stack on a normal pseudomanifold is uniquely defined, and can be obtained with a linear-time algorithm relying on a sequence of collapses. Last, we prove that such a watershed is the cut of the unique minimum spanning forest, rooted in the minima of the Morse stack, of the facet graph of the pseudomanifold.}


\maketitle              

\section{Introduction}

Watershed is a fundamental tool in computer vision, since its inception as an algorithm by the school of mathematical morphology \cite{digabel1978iterative,vincent1991watersheds}. It is still true in this era of deep learning, where it is used as a post-processing tool \cite{arbelaez2010contour}. From a discrete, theoretical point of view, the first topologically-sound approach was proposed in \cite{couprie1997topological,bertrand2005topological}. Building on those results, in~\cite{cousty2009watershed,cousty2009collapses,cousty2014collapses}, it is demonstrated that watersheds are included in skeletons on pseudomanifolds of arbitrary dimension. 

In this paper, we continue exploring the links between watershed and topology, in the framework of discrete Morse theory~\cite{forman1995discrete,scoville2019discrete}. Indeed, mathematical morphology~\cite{najman2013mathematical} and discrete Morse theory, although they pursue different objectives, share many similar ideas. In particular, as demonstrated in \cite{boutry2019equivalence,boutry2021equivalence,boutry:hal-03676854},  filtering minima using \emph{morphological dynamics}~\cite{grimaud1992new} in watershed-based image-segmentation, is equivalent to filtering the minima by \emph{persistence}, a fundamental tool from Persistent Homology~\cite{edelsbrunner2008persistent} used for topological data analysis~\cite{tierny2017introduction,munch2017user}.

Although the main ideas of the present paper originate in \cite{10.1007/978-3-031-19897-7_4}, we have worked towards a simpler, unifying framework for exposing these ideas. This leads us to introduce {\em Morse stacks}: these functions  correspond to the inverse of flat Witten-Morse functions that, according to R. Forman \cite{forman1998witten}, {\em seem to have shown themselves to be the appropriate combinatorial analogue of smooth non-degenerate Morse functions.} We also propose a new definition for normal pseudomanifolds, a class of manifolds on which the different notions of path-connectedness are equivalent. Relying on these notions, we prove that a watershed, a pure $(d-1)$-subcomplex of a normal pseudomanifold, has several interesting properties when defined on a Morse stack $F$. In particular, in this setting, a watershed is uniquely defined, and can be obtained thanks to a linear-time algorithm, relying on a sequence of collapses. Furthermore, a watershed is the cut of the unique minimum spanning forest of the facet graph of the normal pseudomanifold weighted by $F$, rooted in the minima of $F$. Relations between watersheds and Morse theory have long been informally known (see for example \cite{de2013discrete,delgado2014skeletonization,de2015morse} in the discrete setting or \cite{najman1994watershed} in the continuous setting), but this is the first time that a link, relying on a precise definition of the watershed, is presented in the discrete setting. Furthermore, as far as we know, this is the first time that a concept from Discrete Morse Theory is linked to a classical combinatorial optimization problem.

The plan of this paper is the following. Section \ref{sec.maths} provides some basic definitions of simplicial complexes. We introduce here the notion of a covering pair, that is fundamental for the definition of Morse stacks. Section \ref{sec:stacks} recalls some definitions of simplicial stacks, which are a class of weighted simplicial complexes whose upper threshold sets are also complexes. 
Section \ref{sec:normalpseudomanifolds} proposes a new definition for normal pseudomanifolds. Section~\ref{sec:watersheds} provides the necessary definitions for  watersheds on stacks defined on normal pseudomanifolds. We propose here an algorithm for computing watershed relying on the collapse operation. 
In section~\ref{sec:morse}, we introduce Morse stacks. Section~\ref{sec:mosewatersheds} studies the properties of watersheds on Morse stacks. Section~\ref{sec:mst} links  watersheds and minimum spanning forests. 
We conclude the paper with a discussion in section~\ref{sec:discussion}, in which we highlight the importance of our results, from both theoretical and practical points of view, and we propose some perspective for future work.

Finally, appendix \ref{app:pseudo} shows that our definition of normal pseudomanifold is equivalent to the classical one, and appendix \ref{app:DMF} demonstrates that Morse stacks are equivalent to classical discrete Morse functions.

The results of this paper are presented in the setting of simplicial complexes. It should be noted that they can be directly extended to other complexes. In particular, Morse stacks and Morse watersheds, presented hereafter, have a direct equivalent in the class of cubical complexes, 
 which holds significance in the field of image analysis.

\section{Simplicial complexes}
\label{sec.maths}

A \emph{simplex $x$} is a non-empty finite set; the dimension of $x$, written $dim(x)$,
is the number of its elements minus one. We also say that $x$ is a \emph{$p$-simplex} if $dim(x) = p$. 
    
Let $S$ be a finite set of simplexes.
A $p$-simplex in $S$ is a \emph{($p$-)face of $S$}.
A {\em ($p$-)facet of $S$} is a $p$-face of $S$ that is maximal for inclusion.
If $x$ and $y$ are two distinct faces of $S$ such that $x \subseteq y$, we say that \emph{$x$
is a face of $y$ (in $S$)}.
The {\it simplicial closure of $S$} is the set
$S^- = \{ y \subseteq x \; \mid \; y \not= \emptyset$ and $x \in S \}$.
The set $S$ is a {\it (simplicial) complex} if $S = S^-$.

Let $X$ be a complex.
The {\it dimension of $X$}, written $dim(X)$,
is the largest dimension of its simplices,
the {\it dimension of $\emptyset$} being defined to be $-1$. 
\newline
A $0$-face of $X$  is \emph{a vertex of $X$}, and a $1$-face of $X$ is \emph{an edge of $X$}.
\newline
A complex $X$ is a \emph{graph} if the dimension of $X$ is at most $1$. 

Let $X$ be a complex and let $S \subseteq X$. If $S$ is a complex, we say that $S$ is \emph{closed for $X$} or that $S$ is a \emph{subcomplex of $X$}.
We say that $S$ is \emph{open for $X$} or that $S$ is an \emph{open subset of $X$} if, for any $x \in S$, we have $y \in S$ whenever $x \subseteq y$ and $y \in X$.
If $X$ is a complex and $S \subseteq X$, we note that $S$ is closed for $X$ if and only if $X \setminus S$ is open for $X$.
In particular, $\emptyset$ and $X$ are both closed and open for $X$. 

\begin{remark}
The above definitions of open and closed sets correspond to an Alexandov topology \cite{alexandroff1937diskrete}.
It should be noted that the usual topology associated with a simplicial complex $X$ is the \emph{face poset $P$ of $X$},
that is, the set of faces of $X$ ordered by inclusion \cite{barmak2008simple}. The poset $P$ may be seen as the barycentric
subdivision of $X$. In contrast, in this paper, we consider directly the collection of closed sets
that are the subcomplexes of $X$.  
\end{remark}

Let $S$ be a finite set of simplexes.
Let $\pi = \langle x_0, \ldots, x_k \rangle$ be a sequence of elements of $S$. The sequence $\pi$ is a \emph{path in $S$ (from $x_0$ to $x_k$)}  if, for any $i \in [0,k-1]$,
either $x_{i} \subseteq x_{i+1}$ or  $x_{i+1} \subseteq x_{i}$.
We say that \emph{$S$ is connected} if, for any $x,y \in S$, there exists a path from $x$ to $y$ in $S$.
If $S \neq \emptyset$, we say that $T \subseteq S$ is a \emph{connected component of $S$} if $T$ is connected and maximal, with respect to set inclusion, for this property.

\begin{remark} \label{rem:path}
We observe that: 
\begin{itemize}
    \item If $X$ is a complex, then $X$ is connected if and only if, for any vertices $x,y \in X$, there exists a sequence $\langle x = x_0, \ldots, x_k = y \rangle$
of vertices of $X$ such that, for any $i \in [0,k-1]$, $x_i \not= x_{i+1}$ and $x_{i} \cup x_{i+1}$ is an edge of $X$. 
\item If $S$ is an open subset of a complex $X$, then $S$ is connected if and only if, for any facets $x,y$ of $S$, there exists a sequence $\langle x = x_0, \ldots, x_k = y \rangle$
of facets of $S$ such that, for any $i \in [0,k-1]$, $x_{i} \cap x_{i+1}$ is in $S$.
\end{itemize}
\end{remark}

The following simple definition of a covering pair (or a $p$-pair) will play an important role in the sequel of the paper: 
\begin{itemize}
    \item It will first allow us to define a free pair (Definition \ref{def:intro2}).
This corresponds to the operation of collapse
of a simplicial complex introduced by J.H.C. Whitehead \cite{whitehead1939simplicial},
which is a discrete analogue of a retraction, that is, a
continuous (homotopic) deformation of an object onto
one of its subsets. In Section \ref{sec:stacks}, free pairs for a simplicial complex will be extended to free pairs on stacks, which are maps on simplicial complexes (Definition \ref{def:stack1}). 
\item In Section \ref{sec:morse}, we introduce the notion of a flat pair, which is a special case of a covering pair.
This permits us to have a very simple and concise presentation of a Morse stack (Definition \ref{def:morse1}). Indeed,
a basic  link between Morse stacks and the collapse operation is straightforward (Proposition \ref{pro:morse1}).
Furthermore, the notions of a gradient vector field and a gradient path follow immediately from 
covering and flat pairs.  
\end{itemize}

\begin{definition}[Covering pair] \label{def:intro1}
Let $X$ be a complex and $x,y \in X$, with $dim(y) = p$. We say that $(x,y)$ is a \emph{covering pair of $X$} or a
\emph{$p$-pair of $X$} if $x$ is a face of $y$ and $dim(x) = p-1$.
\end{definition}

\begin{definition}[Free pair] \label{def:intro2}
Let $X$ be a complex and let $(x,y)$ be a $p$-pair of $X$.
We say that  $(x,y)$ is a \emph{free ($p$-)pair of $X$} if $y$ is the only face of $X$ that contains~$x$.
\end{definition}

Thus, if $(x,y)$ is a free pair of $X$, we have necessarily $dim(x) = dim(y) - 1$.
Furthermore, we observe that $y$ is necessarily a facet of $X$.

If $(x,y)$ is a free $p$-pair of a complex $X$, then $Y = X \setminus \{x,y \}$
is {\em an elementary ($p$-)collapse of $X$}.
We say that
$X$ {\em collapses (resp. $p$-collapses) onto $Y$},
if there exists a sequence
$\langle X=X_0,...,X_k = Y \rangle$ such that
$X_i$ is an elementary collapse (resp. elementary $p$-collapse) of $X_{i-1}$, $i \in [1,k]$.
If, furthermore, $Y$ has no free pair (resp. free $p$-pair), then $Y$ is an {\em ultimate collapse (resp. ultimate $p$-collapse) of $X$}.
A complex
$X$ is {\em collapsible} if $X$ collapses onto a single vertex.

\section{Simplicial stacks} \label{sec:stacks}

Let $X$ be a simplicial complex, and let $F$ be a map from $X$ to $\bb{Z}$. If $x$ is a face of $X$, the value $F(x)$
is called the \emph{altitude} of $F$ at $x$. For any $\s \in \bb{Z}$, we write $F[\s] = \{x \in X \; \mid \; F(x) \geq \s \}$,
$F[\s]$ is the \emph{$\s$-section} of $F$.
We say that $F$ is a \emph{(simplicial) stack} on $X$ if
any $\s$-section of $F$ is a simplicial complex. In other words, any $\s$-section of $F$ is a closed set for $X$. 

Let $F$ be a map from a complex $X$ to $\bb{Z}$. It may be easily seen that $F$ is a simplicial stack if and only if,
for any $x,y \in X$ such that $x \subseteq y$, we have $F(x) \geq F(y)$. Also, a map $F$ is a simplicial stack if and only if, for any  covering pair $(x, y)$ in~$X$,
we have $F(x) \geq F(y)$. 

Now, we extend the notion of free pairs of simplicial complexes to simplicial stacks.
This extension allows us to introduce some fundamental discrete homotopic transforms of these maps.


\begin{definition} \label{def:stack1}
Let $F$ be a simplicial stack on $X$. We set $\s_m = min \{F(x) \; \mid \; x \in X \}$. Let $(x, y)$ be a $p$-pair of~$X$. We say that
$(x,y)$ is a \emph{free ($p$-)pair of $F$} if $(x,y)$ is a free ($p$-)pair of $F[\s]$, with $\s = F(x)$ and $\s > \s_m$.
\end{definition}

If $(x,y)$ is a free pair of $F$, then both $x$ and $y$ are in $F[\s]$, with $\s = F(x)$. Thus $F(y) \geq F(x)$.
Also, we have $x \subseteq y$. Then, since $F$ is a stack, we have $F(y) \leq F(x)$.
Thus, we have $F(x) = F(y)$ whenever $(x,y)$ is a free pair of~$F$. 

\begin{figure}
    \centering
    \includegraphics[width=.8\columnwidth]{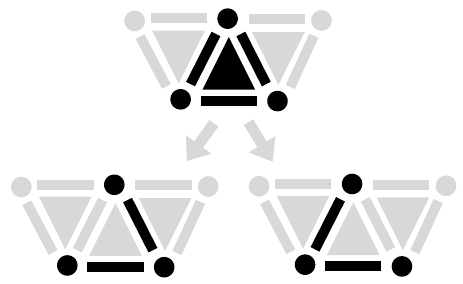}
    \caption{A simplicial stack $F$ with two values (0 in gray, and 1 in black) on a complex $X$, and two different elementary collapses of $F$.}
    \label{fig:collapse}
\end{figure}

Let $F$ be a simplicial stack on a complex $X$. 
\begin{enumerate}
    \item Let $(x,y)$ be a free ($p$-)pair of $F$.
Let $G$ be the map such that $G(z) = F(z) -1$ if $z = x$ or $z = y$,
and $G(z) = F(z)$ if $z \in X \setminus \{x,y\}$. We can see that $G$ is a simplicial stack on $X$. The map $G$ is called an \emph{elementary ($p$-)collapse of $F$
through $(x,y)$}, or, simply, an \emph{elementary ($p$-)collapse of $F$}. Fig.~\ref{fig:collapse} shows two different elementary collapses of a stack $F$.
\item If $G$ is the result of a sequence of elementary collapses (resp. $p$-collapses) of $F$, then we say that $F$
\emph{collapses (resp. $p$-collapses)} onto $G$. 
\item If $F$ collapses (resp. $p$-collapses) onto a stack $G$ that has no free pair (resp. no free $p$-pair),
then $G$ is an \emph{ultimate collapse (resp. ultimate $p$-collapse)} of~$F$.
\end{enumerate}

We conclude this section by giving a definition of a (regional) minimum of a stack, which plays a
crucial role in the notion of a watershed.

Let $F$ be a simplicial stack on a complex $X$ and let $\s \in \bb{Z}$. 
A subset $A$ of $X$ is a \emph{minimum of $F$ (at altitude $\s$)}
if $A$ is a connected component of
$X \setminus F[\s +1]$ and $A \cap (X \setminus F[\s]) = \emptyset$.
The \emph{divide of $F$} is the set composed of all faces of $X$ that are not in a minimum of $F$.
Note that any minimum of $F$  is an open set for $X$, and the divide of $F$ is a simplicial complex.

\section{Normal pseudomanifolds}
\label{sec:normalpseudomanifolds}

The results of this paper hold true in a
large family of $n$-dimensional discrete spaces, namely
the normal pseudomanifolds. This section provides a presentation of these spaces.

Let $S$ be a finite set of simplexes.  A \emph{ strong $p$-path in $S$ (from $x_0$ to $x_k$)} is a path
$\langle x_0,...,x_k \rangle$ such that, for each $i \in [0,k-1]$,
either $(x_i,x_{i+1})$ is a $p$-pair, or $(x_{i+1}, x_i)$ is a $p$-pair.
The set $S$ is \emph{($d$-)pure} if all facets of~$S$ have the same dimension $d$.
If $S$ is $d$-pure, we say that a strong $d$-path in $S$ is a \emph{strong path in $S$}. Also, we say that $S$
 is \emph{strongly connected} if, for any two facets $x,y$ in $S$, there exists a strong path in $S$ from $x$ to $y$. A subset $T$ of $S$ is a \emph{strong connected component of $S$} if $T$ is strongly connected and maximal, with respect to set inclusion, for this property.

\begin{definition}[Normal pseudomanifold] \label{def:normal}
A connected and $d$-pure complex $X$, with $d \geq 1$, is a \emph{normal pseudomanifold} (or a \emph{normal $d$-pseudomanifold}) if: 
\begin{enumerate}
    \item The complex $X$ is \emph{non-branching}, that is, each $(d - 1)-$face of $X$ is included in exactly two $d$-faces of $X$. 
    \item The complex $X$ is \emph{strictly connected}, that is, each connected open subset of~$X$ is strongly connected.
\end{enumerate}
\end{definition}

Recall that a pure complex $X$ is a \emph{pseudomanifold} if it is non-branching and strongly connected \cite{massey91}.
Since the very set $X$ is open for a complex $X$, we see that any normal pseudomanifold is a pseudomanifold. 
\newline
In fact, the above  definition is a new definition for a normal pseudomanifold.
In Appendix \ref{app:pseudo}, we show that it is equivalent to the classical definition
\cite{BAGCHI2008327}, \cite{BASAK2020107035}, \cite{DBNN2008},
which consists of a local condition together with the conditions that must be satisfied by a pseudomanifold.

\begin{figure*}
\begin{center}
 \begin{tabular}{ccc}
   \includegraphics[width=.29\textwidth]{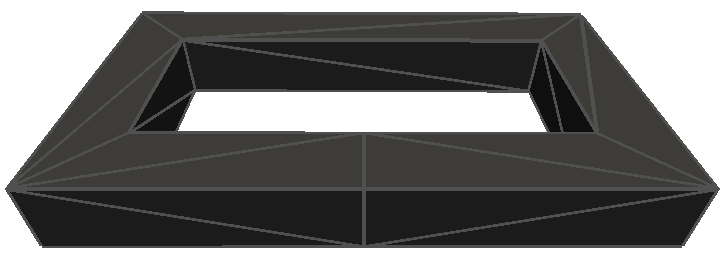} &
   \includegraphics[width=.29\textwidth]{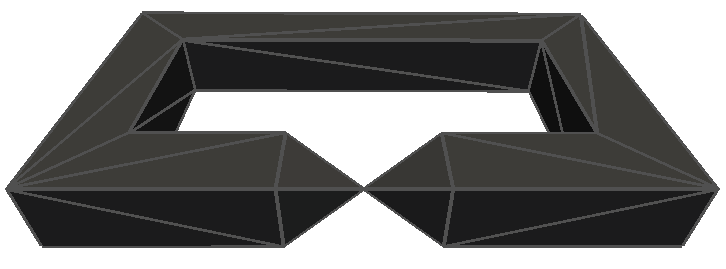} &
   \includegraphics[width=.29\textwidth]{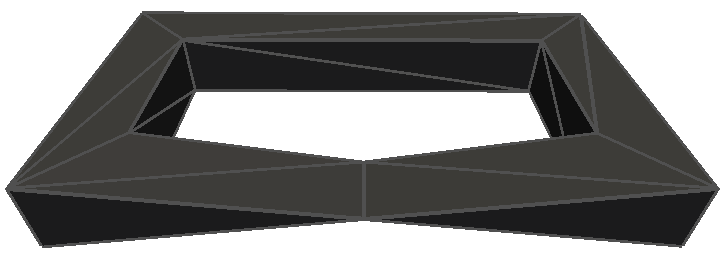} \\
   (a) & (b) & (c)
 \end{tabular}
  \caption{\label{fig:pseudo1}
   (a): A normal pseudomanifold, which is a torus,
   (b): A pseudomanifold, which is a pinched torus, and where the pinch face is a vertex,
   (c): A pinched torus where the pinch face is a segment. This is not a pseudomanifold.
 }
\end{center}
\end{figure*}

Let us consider Fig. \ref{fig:pseudo1}. The triangulated  torus (a) is a normal pseudomanifold. The triangulated pinched torus (b) is a pseudomanifold that
is not normal: the set of all faces containing the pinch vertex is a connected open subset of the complex,
but this set is not strongly connected.
The triangulated pinched torus (c) is not a pseudomanifold: the pinch segment does not satisfy the non-branching condition.

Let $X$ be a proper subcomplex of a $d$-pseudomanifold $M$. We can see that, if $dim(X) = d$,
the complex $X$ has necessarily a boundary, that is,
there exists a free $d$-pair for $X$. By induction, it means that the dimension of an ultimate
$d$-collapse of $X$ is necessarily $d-1$. See \cite{cousty2009collapses} for a formal proof. 

\bigskip

{\bf Important notations.} In the sequel of the paper: 
\begin{itemize}
    \item We denote by $\bb{S}$ the collection of all simplicial complexes.
    \item If $X \in \bb{S}$, we write $Y \preceq X$  whenever  $Y \subseteq X$ and $Y \in \bb{S}$, that is, whenever $Y$ is a subcomplex of $X$. 
    \item If $X \in \bb{S}$, and $S \subseteq X$, we write $S \sqsubseteq X$ whenever $S$ is an open subset of $X$.
    \item We denote by $\bb{M}$ (resp. $\bb{M}_d$) the collection of all normal pseudomanifolds (resp. all normal $d$-pseudomanifolds).
    \item If $F$ is a stack on $M \in \bb{M}$, the notation $\mathfrak{min}(F)$ stands for the union of all minima of $F$, and we write  $\mathfrak{div}(F)$ for the divide of $F$. Thus, we have $\mathfrak{div}(F) = M \setminus \mathfrak{min}(F)$, $\mathfrak{div}(F) \preceq M$, and $\mathfrak{min}(F) \sqsubseteq M$.
\end{itemize}

We are now ready to introduce the notion of a watershed in the context of simplicial complexes.
We will consider a normal pseudomanifold $M \in \bb{M}$ and a simplicial stack $F$ on $M$, the map
$F$ may be seen as a ``topographical relief'' on the space $M$.
A simplicial complex $W \preceq M$ may be a watershed of $F$ if $W$ ``separates the minima of $F$''.
It means that if $A \sqsubseteq M$ is a connected component of $M \setminus W$, then $A$  contains one and only one
set $B \sqsubseteq M$ that is a minimum of $F$.
Furthermore, the complex $W$ must satisfy a ``drop of water principle'': from each face of $W$, we may reach at least two distinct minima
of $F$ by following a descending path.
Each connected component $A$ of $M \setminus W$ will correspond to a ``catchment basin'' of the map $F$.

\section{Watersheds}
\label{sec:watersheds}

Let $X \in \bb{S}$, and let $A \sqsubseteq X$, with $A \not= \emptyset$. 
We say that $B \sqsubseteq X$ is an \emph{extension of $A$} if $A \subseteq B$, and if each connected component of $B$ includes exactly one connected component of $A$. We also say that $B$ is an extension of~$A$ if $A = B = \emptyset$.

 \begin{proposition} \label{pro:water0}
Let $M \in \bb{M}$ and let $A \sqsubseteq M$. 
\begin{enumerate}
\item A subset $S$ of $A$ is a connected component of $A$ if and only if $S$ is a strong connected component of $A$.
 \item Let $B \sqsubseteq M$, with $A \subseteq B$.
The set $B$ is an extension of $A$ if  and only if  each strong connected component of $B$
includes exactly one strong connected component of $A$.
\end{enumerate}
 \end{proposition}

\begin{proof}
 1) It may be seen that, if $S$ is a connected component of the open set $A$, then $S$ is necessarily an open set for $M$.
 Since $M$ is a normal pseudomanifold, we deduce that $S$ is strongly connected. Furthermore, $S$ is a strong connected component of $A$,
 otherwise $S$ would not be a maximal connected subset of $A$. Now, if $S$ is a strong connected component of $A$,
 then $S$ is a connected subset of $A$. Again, we see that $S$ is a connected component of $A$. Otherwise, $S$
 would be a proper subset of a connected open subset $T$ of $A$. Since $M$ is a normal pseudomanifold, this subset $T$ would be strongly connected, and $S$ would not be a maximal strongly connected subset of $A$. \\
 2) is a direct consequence of 1). 
 \end{proof}

Let $X \in \bb{S}$ and $Y \preceq X$. Let $A \sqsubseteq X$, with $A \not=\emptyset$.
 We say that $Y$ is a \emph{cut for $A$}, if $X \setminus Y$ is an extension of $A$, and if $Y$ is minimal for this property.
 That is, if $Z \preceq Y$, and if $X \setminus Z$ is an extension of $A$, then we have necessarily $Z = Y$.

 \begin{proposition} [from {\cite{cousty2009collapses}}] \label{pro:water2}
Let $M \in \bb{M}_d$, $A \sqsubseteq M$ and $X \preceq M$, with $A \not= \emptyset$.
If $X$ is a cut for $A$, then the complex $X$ is either empty
or a pure $(d - 1)$-complex.
 \end{proposition}

\begin{remark}
It could be seen that the previous result no longer holds if we consider arbitrary pseudomanifolds instead of normal pseudomanifolds.
For example, the pinched vertex of the pinched torus of Figure \ref{fig:pseudo1}.(b) could be in a cut. 
\newline
In fact, it is possible to bypass this situation by considering only strong paths between faces, as it is done in~\cite{cousty2009collapses}.
In this paper, in order to handle general connectedness and arbitrary paths, we have made the choice to settle our results
in normal pseudomanifolds.
\end{remark}

Let $F$ be a stack on $M \in \bb{M}$.
If $\pi = \langle x_0,\dots,x_{k} \rangle$ is a path in $M$, we say that $\pi$ is
\emph{ascending for $F$} (resp. \emph{descending for $F$}) if, for any $i \in [0,k-1]$,
we have $F(x_i) \leq F(x_{i+1})$ (resp.  $F(x_i) \geq F(x_{i+1})$).

\begin{definition}[Watershed] \label{def:water3}
Let $F$ be a stack on $M \in \bb{M}$ and let $X \preceq M$ be a cut for $\mathfrak{min}(F)$.
We say that~$X$ is a \emph{watershed} of $F$ if, for each $x \in X$, there exist two strong paths $\pi_1 = \langle x_0, \dots,x_k \rangle$
and $\pi_2 = \langle y_0, \dots, y_l \rangle$ in $M \setminus X$, such that: 
\begin{itemize}
    \item $x \subseteq x_0$ and $x \subseteq y_0$;
    \item $\pi_1$ and $\pi_2$ are descending paths for $F$; and
    \item $x_k$ and $y_l$ are simplices of two distinct minima of $F$.
\end{itemize}
\end{definition}

Let $M \in \bb{M}$, $F$ be a stack on $M$, and let $W$ be a watershed for $F$.
We say that $B \sqsubseteq M$ is a \emph{(catchment) basin of $W$} if $B$ is a connected component of
$\overline{W} = M \setminus W$. Since $W$ is a cut for $\mathfrak{min}(F)$, 
\begin{itemize}
    \item any catchment basin $B$ of $W$ contains a unique minimum $A$ of $F$, we say that $A$ is the 
    \emph{minimum of $B$}; 
    \item any minimum $A$ of $F$ is included in a unique basin $B$ of $W$, we say that $B$ is the 
    \emph{catchment basin of $A$}.
\end{itemize}

\begin{proposition} [from {\cite{cousty2009collapses}}] \label{pro:water4}
Let $M \in \bb{M}_d$, $F$ be a stack on $M$, and $W$ be a watershed of $F$.
Then, for any $d$-face $x$ in $M$, there exists a strong path in $M \setminus W$ from $x$
to a $d$-face of a minimum of $F$, that is descending for $F$.
\end{proposition}

From the previous result, we easily derive the following proposition.

\begin{proposition} \label{pro:water5}
Let $M \in \bb{M}$, $F$ be a stack on $M$, and $W$ be a watershed of $F$.
Let $B$ be the catchment basin of a minimum $A$ of $F$. Then, for any $x \in B$, there exists a descending path in $B$ from $x$
to a face of $A$.
\end{proposition}

The two following results are crucial for linking a watershed of a stack $F$ and the homotopy of $F$.

\begin{proposition} [from {\cite{cousty2009collapses}}] \label{pro:water6}
Let $M \in \bb{M}_d$. 
If $F$ is a stack on $M$ and $H$ is a collapse of $F$, then: 
\begin{enumerate}
    \item $\mathfrak{min}(H)$ is an extension of $\mathfrak{min}(F)$.
    \item $\mathfrak{div}(H)$ is a collapse of $\mathfrak{div}(F)$.
\end{enumerate}
\end{proposition}

\begin{figure}
\begin{center}
 \begin{tabular}{cc}
   \includegraphics[width=.2\textwidth]{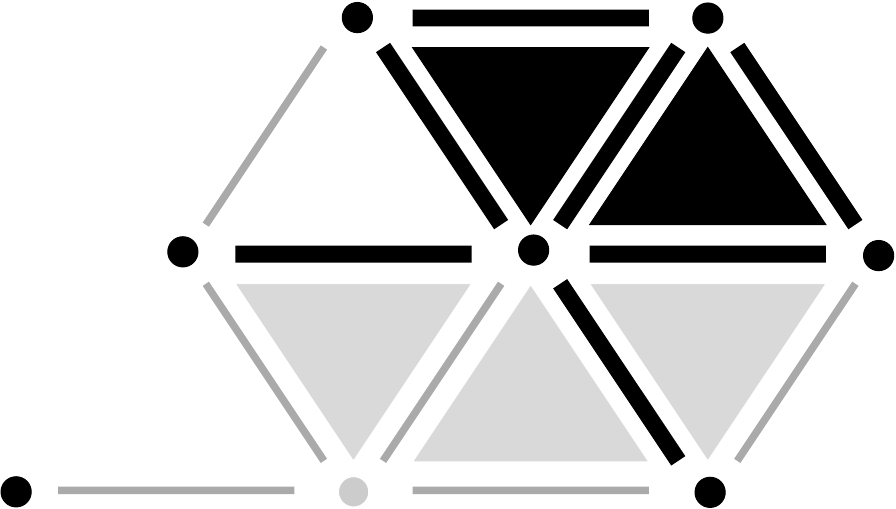} &
   \includegraphics[width=.2\textwidth]{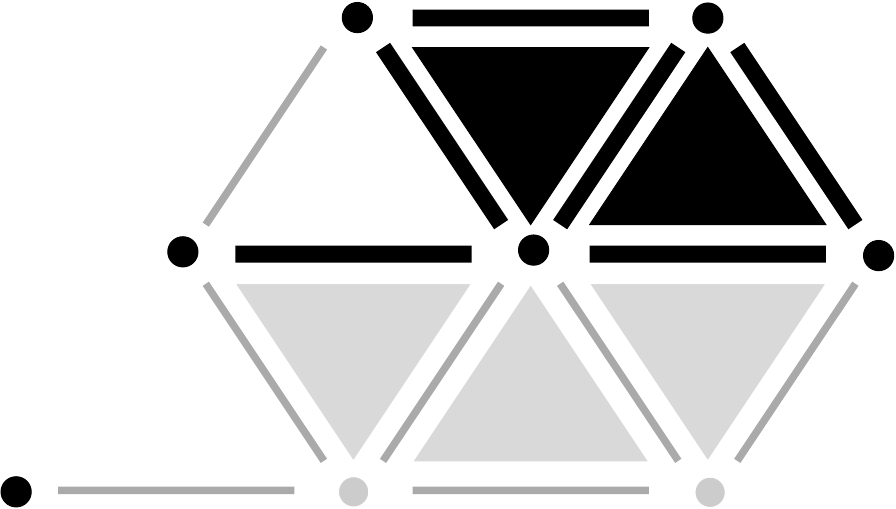} \\
   F &  H
 \end{tabular}
  \caption{\label{fig:stackA}
  Two stacks $F$ and $H$ with two levels: altitude $0$ (faces in light gray) and altitude $1$ (faces in black).
  The stack $H$ is an elementary collapse of $F$ (at altitude $1$). But $F$ has three minima whereas $H$ has only two.
  Thus, $\mathfrak{min}(H)$ is not an extension of $\mathfrak{min}(F)$.
}
\end{center}
\end{figure}

It should be noted that the previous proposition is no longer true if we consider a stack on an arbitrary  complex $X \in \bb{S}$ rather than a complex $M \in \bb{M}$.
See Fig.~\ref{fig:stackA}, which provides a simple counter-example.

\begin{proposition} [from {\cite{cousty2009collapses}}] \label{pro:water7}
Let $M \in \bb{M}_d$ and $F$ be a stack on $M$. 
\begin{enumerate}
    \item $F$ contains a free $d$-pair if and only if $\mathfrak{div}(F)$ contains a free $d$-pair. 
    \item If $dim(\mathfrak{div}(F)) = d$, then there exists a free $d$-pair for $F$.
\end{enumerate}
\end{proposition}

\begin{figure}
    \centering
    \includegraphics[width=.99\columnwidth]{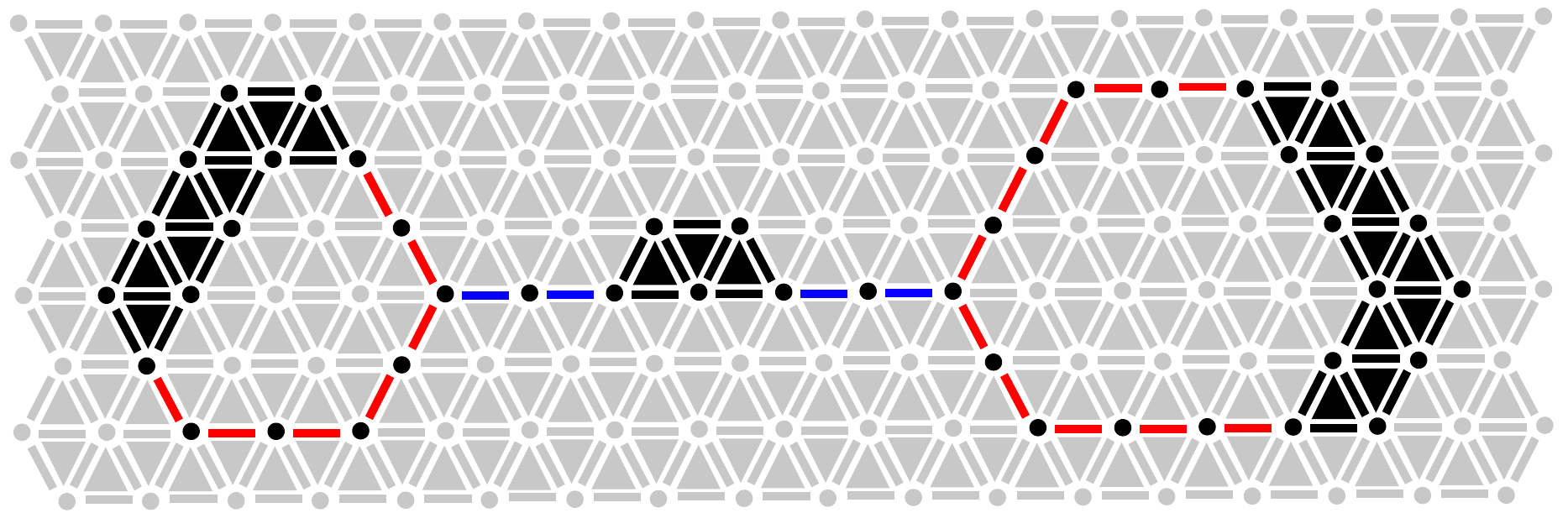}
    \caption{A simplicial stack $F$ with two values (0 in gray and 1 in black, red, and blue) on a subset of a normal 2-pseudomanifold. In blue and red, the separating faces. In red, the biconnected faces.}
    \label{fig:separating}
\end{figure}

Let $F$ be a stack on $M \in \bb{M}_d$ and $x$ be a $(d-1)$-face of $M$.
Let $y$, $z$ be the two $d$-faces containing $x$. We say that $x$
 is \emph{(locally) separating for $F$} if $F(y) < F(x)$ and $F(z) < F(x)$.
 We say that $x$ is \emph{biconnected for $F$}
if $y$ and $z$ belong to distinct minima of $F$. See Fig.~\ref{fig:separating} for an illustration of these two notions.

\begin{definition} \label{def:water4}
Let $F$ be a stack on $M \in \bb{M}_d$. Let $X \preceq M$.
We say that $X$ is a \emph{cut by collapse of $F$}, or a
\emph{$\mathcal{C}$-watershed of $F$},  if there
exists an ultimate $d$-collapse~$H$ of $F$ such that $X$ is the
simplicial closure of the set of all faces of $M$ that are
biconnected for $H$.
\end{definition}

\begin{figure}
\begin{center}
\begin{tabular}{c}
    \includegraphics[height=.2\textwidth]{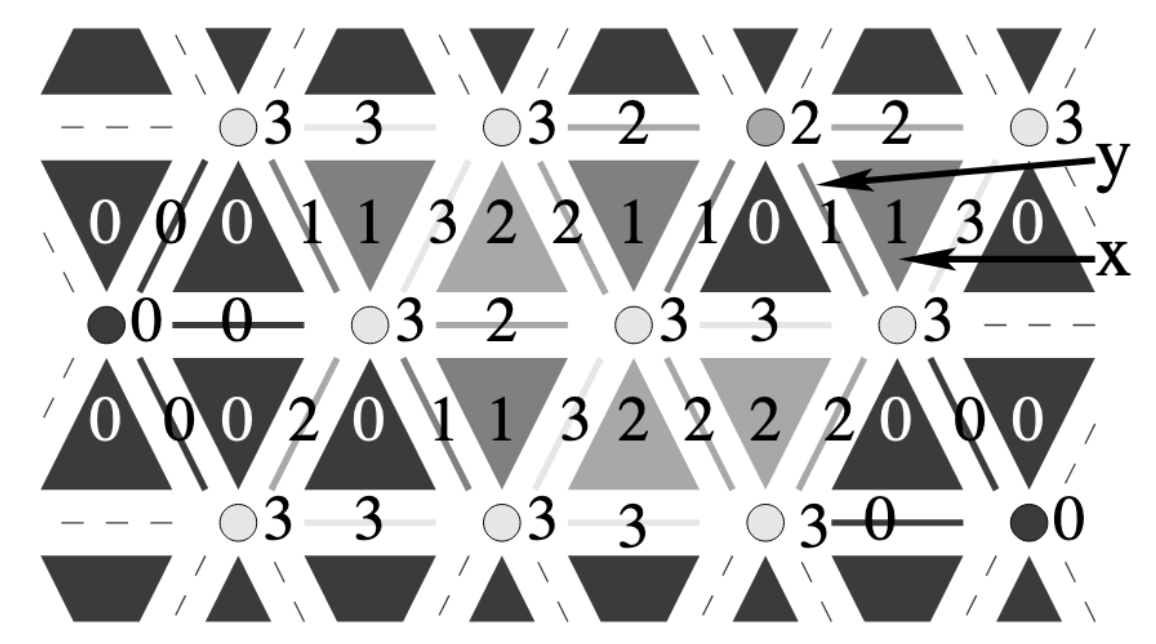} \\
    (a)\\
    \includegraphics[height=.2\textwidth]{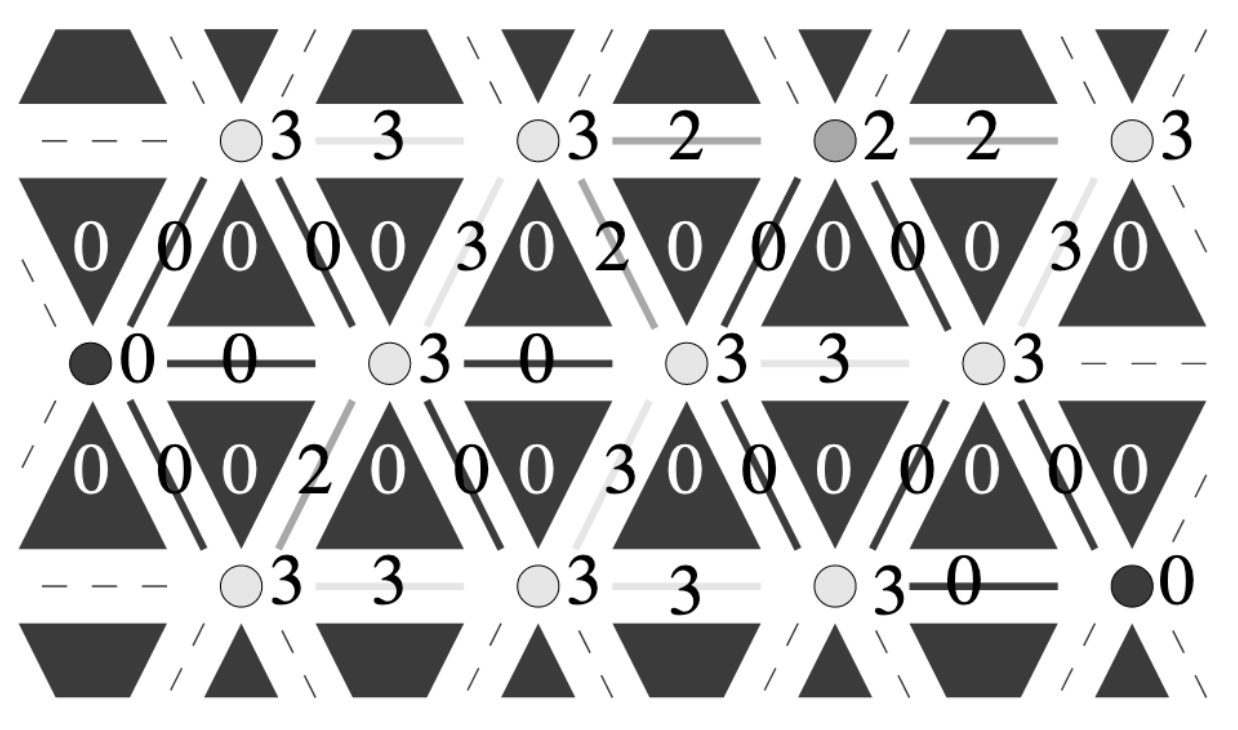} \\
    (b) \\
    \includegraphics[height=.2\textwidth]{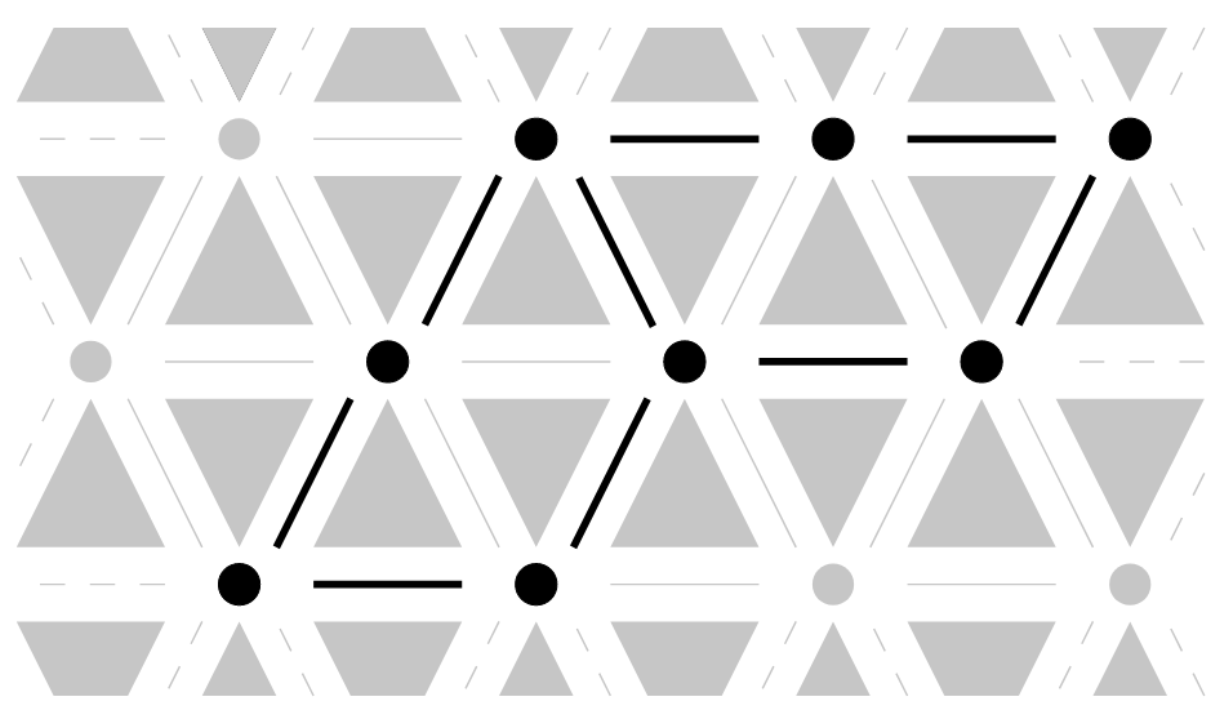} \\
   (c)\\
 \end{tabular}
  \caption{\label{fig:watershedOnSimplicialStack}
    (a) A simplicial stack $F$ on a subset of a normal 2-pseudomanifold. (b) An ultimate 2-collapse of $F$. (c) A watershed of $F$. The pair $(y,x)$ in (a) is a free-pair for $F$.
  }
\end{center}
\end{figure}

\begin{theorem} [from {\cite{cousty2009collapses}}] \label{th:water1}
Let $M \in \bb{M}$ and let $F$ be a stack on $M$.
A complex $X \preceq M$ is a watershed of $F$ if and only if $X$ is a $\mathcal{C}$-watershed of $F$.
\end{theorem}

Theorem \ref{th:water1} is illustrated in Fig.~\ref{fig:watershedOnSimplicialStack}.


\begin{procedure*}[tb]
 \KwData{A stack $F$ defined on a normal pseudomanifold $M$}
 \KwResult{A watershed $W$ of $F$}
 Set $H=F$;
 
\nl \Repeat{$H$ has no free $d$-pair}{Select arbitrarily a free $d$-pair $(x,y)$ of $H$;
 
 and replace~$H$ by the elementary collapse of $H$ through $(x,y)$;}

\nl Label all $d$-faces of distinct minima of $H$ with distinct labels;

\nl Extract from $H$ the complex $W$ that is the simplicial closure of the set of all $(d-1)$-faces of~$M$ that are biconnected for $H$. 
\caption{WatershedCollapse($F$,$M$) -- computes a watershed $W$ of a stack $F$ defined on a normal pseudomanifold $M$.}
\label{alg:WatershedCollapseAlgo}
\end{procedure*}

From Theorem \ref{th:water1}, we may derive the  procedure 
\ref{alg:WatershedCollapseAlgo} for obtaining a watershed of 
a stack $F$ on $M \in \bb{M}_d$.

The result $W$ depends on the choices of the free pairs that are made at step~1. In any case,
any watershed of $F$ may be obtained by this procedure.

As explained hereafter, a direct implementation of the algorithm \ref{alg:WatershedCollapseAlgo}
can be slow.

\begin{itemize}
    \item Step 1 is the more complex one. A naive implementation of this step is in the order of $n^2*h$, where $n$ is the number of $d$-faces, and $h$ is the number of different altitudes of $F$. However, this step can be done in quasi-linear time, relying on a straightforward adaptation to simplicial complexes of the algorithm presented in \cite{10.1007/978-3-540-31965-8_17}. This algorithm relies on a tree structure, where the nodes of the tree are the connected components of all the level sets of $F$, and where the edges of the tree correspond to the parenthood relationships between those connected components.
    \item Step 2 is a simple labelling, and may be done in linear time with respect to the number of $d$-faces. By using such a labelling, checking if a $d$-face is biconnected can be done in constant time.
    \item Finally, step 3 may be implemented in linear time, with respect to the number of incidence relations of $M$, that is the cardinality of the set  $\{(x,y) \; \mid \; x, y \in M$ and $x \subsetneq y \}$. 
\end{itemize}

\section{Morse stacks} \label{sec:morse}

In this section, we transpose some basic notions of discrete Morse theory to stacks.
We proceed by defining a Morse stack, which is the counterpart of a classical discrete Morse function.
Morse stacks simply correspond to the inverse of flat discrete Morse functions. Since any discrete Morse function is equivalent to a flat discrete Morse function, there is no loss of generality to develop our notions with Morse stack.
See Appendix \ref{app:DMF}, which provides some properties linking these notions. 

Let $F$ be a map from a complex $X$ to $\bb{Z}$.
We say that a covering pair $(x,y)$ of~$X$ is a \emph{flat pair of $F$} whenever we have $F(x) = F(y)$.

\begin{definition}[Morse stack]\label{def:morse1}
Let $F$ be a simplicial stack on a complex $X$.
We say that $F$ is a \emph{Morse stack (on $X$)} if any face of $X$ is in at most one flat pair of~$F$.
\end{definition}

Let $F$ be an arbitrary simplicial stack, and let $(x,y)$ be a covering pair of~$F$. We have seen that, if $(x,y)$ is a free pair of $F$, then
necessarily $(x,y)$ is a flat pair of $F$. Suppose now that $(x,y)$ is a flat pair of $F$. Then, there may exist another covering pair $(x,z)$
that is also a flat pair of $F$. In this case, we see that $(x,y)$ is not a free pair of $F$. By the very definition of a Morse stack, this situation cannot occur.
In fact, we have the following result.

\begin{Proposition} \label{pro:morse1}
Let $F$ be a Morse stack on a complex $X$. A covering pair $(x,y)$ of~$X$ is a free pair of $F$ if and only if $(x,y)$ is a flat pair of $F$.
\end{Proposition}

\begin{definition}[Regular and critical simplex]
Let $F$ be a Morse stack on a complex $X$ and let $x \in X$ with $dim(x) = p$.
\begin{itemize}
    \item We say that $x$ is \emph{regular} or \emph{$p$-regular for $F$} if $x$ is in a flat pair of $F$.
    \item We say that $x$ is \emph{critical} or \emph{$p$-critical for $F$} if $x$ is not regular for $F$.
\end{itemize}
\end{definition}


Let $F$ be a Morse stack on a complex $X$. 
\newline
The \emph{gradient vector field of $F$}, written
$\gradient(F)$, is the set of all flat pairs of $F$. 
\newline
If $(x,y)$ is a covering pair of $F$ such that $F(x) > F(y)$, we say that
$(y,x)$ is a \emph{differential pair of $F$}. We write
$\dif{F}$ for the set of all differential pairs for $F$. 
\newline
We also set: 
\begin{itemize}
    \item $\gradient_p(F) = \{(x,y) \in \gradient(F) \; \mid \; dim(y) = p \}$, and 
    \item $\dif_p(F) = \{(y,x) \in \dif(F) \; \mid \; dim(y) = p \}$.
\end{itemize}


\begin{figure}
    \centering
    \includegraphics[width=.8\columnwidth]{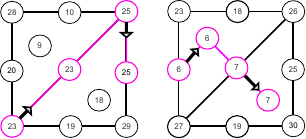}
    \caption{Left: A Morse stack F on a complex X, with, in pink, a $\nabla_1$-path. Right: A Morse stack G on a complex X', with, in pink, a $\nabla_2$-path.}
    \label{fig:lambdapath}
\end{figure}

A \emph{$\nabla_p$-path  in $F$ (from $x_0$ to $x_k$)} is a sequence $\pi = \langle x_0,x_1,...,x_{k} \rangle$ composed of faces of $X$ such that,
for all $i \in [0,k-1]$, the pair $(x_i,x_{i+1})$ is either in $\gradient_p(F)$ or in $\dif_p(F)$. See Fig.~\ref{fig:lambdapath} for an illustration.
A sequence $\pi$ is a \emph{gradient path for $F$} if $\pi$ is a $\nabla_p$-path for $F$ for some $p$.

Let $\pi = \langle x_0,x_1,...,x_{k} \rangle$ be a $\nabla_p$-path in $F$. We observe that: 
\begin{itemize}
    \item For any $i \in [1,k-1]$, the pair $(x_{i},x_{i+1})$ is in $\gradient_p(F)$ (resp. $\dif_p(F)$)
whenever $(x_{i-1},x_{i})$ is in $\dif_p(F)$ (resp. $\gradient_p(F)$). 
    \item Each face of $\pi$ is either a $p$-face or a $(p-1)$-face. For any $i \in [0,k-1]$,
if $x_i$ is a $p$-face, then $x_{i+1}$ is a $(p-1)$-face, and if $x_i$ is a $(p-1)$-face, then $x_{i+1}$ is a $p$-face. 
    \item If $\pi$ is not trivial, then $\pi$ cannot be closed, that is, we have necessarily $k = 0$ whenever $x_k = x_0$.
    \item The path $\pi$ is an ascending path for $F$, that is, we have $F(x_i) \leq F(x_{i+1})$ for any $i \in [0,k-1]$. Furthermore, we have
$F(x_i) < F(x_{i+2})$ for any $i \in [0,k-2]$.
\end{itemize}

As a consequence of the above remarks, we can easily derive the following result, which provides an 
alternative definition of a gradient path.   

\begin{Proposition} \label{prop:gradpath}
Let $F$ be a Morse stack on $X \in \bb{S}$.
A path $\pi$ in $X$ is a ${\nabla}_p$-path in $F$ if and only if $\pi$ is a strong $p$-path in $X$ that is an ascending path for $F$.
\end{Proposition}

\begin{remark} 
The above definition of a gradient path is a simple extension of the classical definition, and also of the definition 
given in \cite[Definiton 2.46]{scoville2019discrete}. Here we allow beginning and ending the sequence
with either a $p$-face or a $(p-1)$-face. 
\end{remark} 

We have seen the importance of the minima of a simplicial stack for the definition of a watershed. The following proposition allows us to give a precision about this notion in the case of Morse stacks.

\begin{Proposition} \label{pro:mw1}
Let $F$ be a Morse stack on a complex $X \in \bb{S}$.
If $S \subseteq X$ is a minimum of $F$, then $S$ contains a single facet $x$ of~$X$.
Furthermore, we have either $S = \{x\}$ or $S = \{ x ,y \}$, where $(y,x)$ is a free pair for $X$.
\end{Proposition}
\begin{proof}
    Let  $S$ be a minimum for $F$, and let $\lambda$ be the altitude of $S$. The set $S$ is open for $X$, thus $S$ contains at least one facet of $X$. 
    Suppose $S$ contains more than one facet. By Remark \ref{rem:path}, since $S$ is connected, there exist two distinct facets $x$, $y$ of $X$ such that $x \in S$, $y \in S$, and $z = x \cap y \in S$. We have $F(x) = F(y) = F(z) =\lambda$. Since $F$ is a stack, we have $F(t) = \lambda$ for all $t$ such that $z \subseteq t \subseteq x$ and all $t$ such that $z \subseteq t \subseteq y$. Also, there exist two distinct faces $x' \subseteq x$, $y' \subseteq y$, such that $(z,x')$ and $(z,y')$ are covering pairs for $X$. But these pairs are flat pairs for $F$. In this case, the face $z$ would belong to more than one flat pair, a contradiction.
    
    Thus $S$ contains a single facet $x$. By the very definition of a Morse stack, it may easily be seen that we have either
    $S = \{x\}$ or $S = \{ x ,y \}$, where $(y,x)$ is a free pair for $X$.
\end{proof}

A pseudo-manifold has no free pairs, thus we have:

\begin{Proposition}
Let $X \in \bb{M}$ and let $F$ be a Morse stack on $X$.
If $S$ is a minimum for $F$, then $S$ is necessarily composed of a single face that is a facet of~$X$.
\end{Proposition}

In the results given in the sequel, we will only consider complexes that are normal pseudo-manifolds. We will say that a face $x \in X$ is \emph{a minimum (of $F$)} whenever the set $\{x\}$ is a minimum of $F$.

\section{Morse stacks and watersheds}
\label{sec:mosewatersheds}

Let $F$ be a Morse stack on a complex $X \in \bb{S}$.
Let $x$, $y$ be two faces of $X$. We say that  \emph{$x$ is $\nabla_p$-linked to $y$}
if there is a $\nabla_p$-path in $F$ from $x$ to $y$.
Let $\pi = \langle x=x_0,...,x_{k}=y \rangle$ be a $\nabla_p$-path in $F$ from $x$ to $y$. We write
$\tilde{\pi} = \langle y=x_k,...,x_{0}=x \rangle$ and we say that $\tilde{\pi}$ is a $\tilde{\nabla}_p$-path in $F$ from $y$ to~$x$.
We say that a face $z \in X$ is \emph{an extension of $\pi$} if
$\langle x=x_0,...,x_{k}=y, z \rangle$ is a $\nabla_p$-path in $F$ from $x$ to $z$.
We say that $z$ is \emph{an extension of $\tilde{\pi}$}
if $\langle y=x_k,...,x_{0}=x, z \rangle$ is a $\tilde{\nabla}_p$-path in $F$ from $y$ to~$z$.

\begin{Proposition} \label{pro:mw2}
Let $F$ be a Morse stack on $M \in \bb{M}_d$, and let $x$ be a facet of $X$.
Let $\tilde{\pi}$ be a $\tilde{\nabla}_d$-path in $F$ from the facet $x$ to a face $y \in X$.
Then one and only one of the following statements is true: 
\begin{enumerate}
    \item The face $y$ is a minimum.
    \item There exists a unique face that is an extension of $\tilde{\pi}$.
\end{enumerate}
\end{Proposition}

\begin{proof}
We set $\tilde{\pi} = \langle x=x_0,...,x_{k}=y \rangle$. 
\newline
i) Suppose $y$ is a $(d-1)$-face. In this case, $y$ cannot be a minimum. Furthermore, 
we have $k \geq 1$ (since $x$ is a $d$-face),
and the face $t = x_{k-1}$ is necessarily a $d$-face.
By the definition of a $\tilde{\nabla}_p$-path, the pair $(y,t)$ is a flat pair.
Since $X$ is a pseudomanifold, there exists a unique $d$-face $z \in X$
such that $t \cap z = y$. We must have $F(z) < F(y)$, otherwise $y$ would belong to more than one flat pair.
Therefore, the pair $(z,y)$ is a differential pair and $z$ is an extension of $\tilde{\pi}$.
Since $(y,t)$ and $(y,z)$ are the only covering pairs that contain $y$, the face $z$ is the unique extension of $\tilde{\pi}$. 
\newline
ii) Suppose $y$ is a $d$-face.
By the definition of a $\tilde{\nabla}_p$-path, a face $z$ is an extension of $\tilde{\pi}$ if and only if $(z,y)$ is a flat pair.
Now we observe that $y$ is not a minimum if and only if there is a face $z$ such that $(z,y)$ is a flat pair.
But $y$ belongs to at most one flat pair. Thus, $\tilde{\pi}$ has a unique extension whenever $y$ is not a minimum. 
\end{proof}
Let $F$ be a Morse stack on $M \in \bb{M}_d$.
It should be noted that, if $\pi$ is  a $\nabla_d$-path in $F$ from a facet $x$ to a face $y$,
then $\pi$ may have more than one extension.
Nevertheless, by induction, we obtain the following result from Prop.~\ref{pro:mw2}.

\begin{Proposition} \label{pro:mw3}
Let $F$ be a Morse stack on $M \in \bb{M}_d$, and let $x$ be a facet of $X$.
There exists a unique minimum $m$ of $F$ such that $m$ is $\nabla_d$-linked to $x$.
Furthermore, there exists a unique $\nabla_d$-path in $F$ from $m$ to $x$.
\end{Proposition}

We now consider the case where a ${\nabla}_d$-path has no extension. Recall that a $(d-1)$-face $x$ is separating for $F$ if the two
$d$-faces $y$, $z$ which contain $x$, are such that $F(y) < F(x)$ and $F(z) < F(x)$.

\begin{Proposition} \label{pro:mw4}
Let $F$ be a Morse stack on $M \in \bb{M}_d$.
Let $x$ be a facet of $X$, and let $\pi$ be a ${\nabla}_d$-path in $F$ from $x$ to a face $y \in M$.
If $\pi$ has no extension, then the face $y$ is necessarily separating for $F$.
\end{Proposition}

\begin{proof}
Let  $\pi = \langle x=x_0,...,x_{k}=y \rangle$ be a ${\nabla}_d$-path in $F$ from $x$ to $y \in M$. 
\newline
If $dim(y) = d$, then $\pi$ has necessarily an extension.
Now suppose  $dim(y) = d-1$, thus $k \geq 1$. The face $y$ is a face of two $d$-faces, the face $z = x_{k-1}$ and another face $t$.
Since $\pi$ has no extension, we have $F(t) < F(y)$.
Furthermore, by the very definition of a ${\nabla}_d$-path, the pair $(z,y)$ is a differential pair,
thus we have $F(z) < F(y)$. Therefore, $y$ is separating for~$F$. 
\end{proof}

Let $F$ be a Morse stack on a complex $X \in \bb{S}$.
Let $\pi$ be a $\nabla_p$-path in $F$. We say that $\pi$ is \emph{maximal} if neither $\pi$ nor $\tilde{\pi}$ has an extension.
The following result is a direct consequence of Prop. \ref{pro:mw2} and \ref{pro:mw4}.

\begin{corollary} \label{cor:mw5}
Let $F$ be a Morse stack on $M \in \bb{M}_d$.
Let $\pi$ be a ${\nabla}_d$-path in $F$ from $x$ to $y$.
If $\pi$ is maximal, then $x$ is a minimum of $F$ and $y$ is a separating face for $F$.
\end{corollary}

Let $F$ be a Morse stack on $M \in \bb{M}_d$. Let $x$ be a $(d-1)$-face of $X$,
and let $y$, $z$ be the two distinct $d$-faces containing $x$.
According to Prop.~\ref{pro:mw3}, each of these faces is $\nabla_d$-linked to a single minimum.
We say that the face $x$ is \emph{$\nabla$-biconnected~(for $F$)} if these two minima are distinct.
Observe that a face is necessarily separating whenever it is $\nabla$-biconnected.

\begin{definition}[Morse watershed]
Let $F$ be a Morse stack on $M \in \bb{M}$.
The \emph{Morse watershed of $F$} is the complex that is the simplicial closure of the set
composed of all faces that are $\nabla$-biconnected for $F$.
\end{definition}

\begin{figure}
\begin{center}

    \includegraphics[width=1\columnwidth]{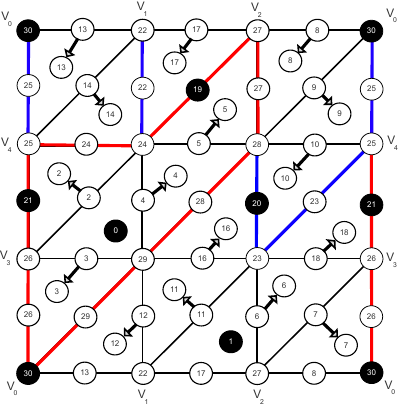} 
  \caption{\label{fig:MorseWatershed}
     Representation of a Morse stack defined on a normal pseudomanifold, a 2D-torus. The separating faces are colored in red and in blue. The red faces are furthermore biconnected, and thus form the Morse watershed.  The critical simplexes are the nodes colored in black; those of dimension 2 are minima, those of dimension 1 are saddle, and those of dimension 0 are maxima. Arrows represent gradient/collapse of dimension 2. We observe that the critical simplex 20 does not belong to the watershed.
  }
\end{center}
\end{figure}

The following theorem~\ref{th:mw1}, illustrated in Fig.~\ref{fig:MorseWatershed}, connects the watershed and the Morse watershed. 

\begin{theorem} \label{th:mw1} Let $F$ be a Morse stack on $M \in \bb{M}$.
The Morse watershed of $F$ is a watershed of $F$.
Furthermore, the Morse watershed of $F$ is the unique watershed of~$F$.
\end{theorem}

\begin{proof}
Let $W$ be the Morse watershed of $F$. 
\begin{enumerate}
    \item
Let $A$ be a connected component of $M \setminus W$. By Prop. \ref{pro:water0}, 
the set $A$ is a strong connected component of $M \setminus W$. 
Let $f$ be a facet of $A$. By Prop. \ref{pro:mw3}, there exists a $\nabla_d$-path $\pi$ from a minimum $m$ of $F$ to $f$.
By the very definition of a $\nabla_d$-path, it may be seen that $\pi$ does not contain any separating face. 
Thus $\pi$ is included in $M \setminus W$, and $m$ is in $A$. Therefore, 
$A$ contains a minimum $m$. \\
Now let $x$ be an arbitrary facet of $A$. Since $A$ is strongly connected, there exists a strong path 
$\pi= \langle m=x_0,...,x_{k}=x \rangle$ in $A$ from $m$ to $x$. By Prop. \ref{pro:mw3}, for each facet of $\pi$,
there is a unique minimum of $F$ that is $\nabla_d$-linked to this facet. \\
Let $x_i$ be a facet of $\pi$, with $i \leq k-2$. 
Thus, $x_{i+1}$ is a $(d-1)$-face and $x_{i+2}$ is a facet. 
Let $m_i$ (resp. $m_{i+2}$) be the unique minimum of $F$ that is $\nabla_d$-linked to $x_i$ (resp. to $x_{i+2}$). 
Since $x_{i+1}$ is not $\nabla$-biconnected for $F$, it may be checked that we have necessarily $m_i = m_{i+2}$. 
By induction, it follows that $m$ is $\nabla_d$-linked to the facet $x$. Since this result holds for any facet of $A$, this
clearly implies that $m$ is the unique minimum of $F$ which is in $A$. 
Thus, any connected component of $M \setminus W$ contains exactly one minimum of $F$. 
But any minimum of $F$ is included in $M \setminus W$ (since a minimum is a $d$-face).
It follows that $M \setminus W$ is an extension of $\mathfrak{min} (F)$.
Furthermore, by the definition of a $\nabla$-biconnected face,
$W$ is minimal for this last property. Therefore, $W$ is a cut for $\mathfrak{min}(F)$.
Since any $\tilde{\nabla}_d$-path is a descending strong path, it may be checked that 
$W$ fulfills all the conditions of Definition \ref{def:water3}: $W$ is a watershed. 
 
    \item Let $W'$ be a watershed of $F$. 
Let $x$ be a $(d-1)$-face of $M$ that is $\nabla$-biconnected for $F$,
and let $y$, $z$ be the two distinct $d$-faces containing $x$. By Prop. \ref{pro:water4}, there exist a descending strong path in $M \setminus W'$ from $y$ to a minimum $m$, and a descending strong path in $M \setminus W'$ from $z$ to a minimum $m'$.
By Prop.~\ref{prop:gradpath}, any descending strong path in $M$ is also a $\tilde{\nabla}_d$-path in~$M$. 
Thus, by the very definition of a $\nabla$-biconnected face, we must have $m \not= m'$. 
Therefore, the face $x$ must be in $W'$, otherwise $y$, $z$, $m$, and $m'$, would belong to the same connected component of  $M \setminus W'$. Thus, $W \subseteq W'$. Since $M \setminus W' \subseteq M \setminus W$, each connected component of $M \setminus W'$ is included in one connected component of $M \setminus W$. But $M \setminus W'$ must be maximal for this last property, otherwise $W'$ would not be a cut for $\mathfrak{min}(F)$. 
It follows that we have $W' = W$.
\end{enumerate}
\end{proof}

By Theorem \ref{th:water1}, the Morse watershed may be obtained by the algorithm
\ref{alg:WatershedCollapseAlgo}.
In this case, we have a greedy procedure, since
the result~$W$ does not depend on the choice of the free pair of $H$ that is made at each iteration.

In fact, since $F$ is a Morse stack, we can simplify this procedure. The  
algorithm 
\ref{alg:MorseWatershedAlgo}
extracts the Morse watershed $W$ of a Morse stack 
$F$ on $M \in \bb{M}_d$. Also, it provides the catchment basin $B$  of each minimum of $F$.


\begin{procedure*}
\KwData{A Morse stack $F$ defined on a normal pseudomanifold $M$}
\KwResult{The Morse watershed $W$ of $F$}
\nl
 Label all faces $x \in M$ with the label $W(x) := False$;

 Label all $d$-faces $x \in \mathfrak{min}(F)$ of distinct minima of $F$ with distinct labels $B(x) \neq 0$;
 
 Label all $d$-faces $x \in  M\setminus\mathfrak{min}(F)$ with the label $B(x) := 0$;

\nl Insert all  $d$-faces $x$ such that $B(x) \neq 0$ in a list $L$;

\nl\Repeat{$L$ is empty}{
Extract  a face $x$ from $L$;

\ForAll{$y$ such that $z = x \cap y$ is a $(d-1)$-face}{
If $F(y) = F(z)$ insert $y$ in $L$ and do $B(y) := B(x)$;

If $B(y) \neq 0$ and $B(y) \not= B(x)$ do $W(z) := True$;
}
}

\nl\lFor{all $(d-1)$-faces $x \in M$ with $W(x) = True$ and all $y \subsetneq x$}{$W(y) := True$}

\nl\lFor{all $d$-faces $x \in M$ and all $y \subsetneq x$ with $W(y) = False$}{$B(y) := B(x)$} 
    \caption{MorseWatershed($F$,$M$) -- computes the Morse watershed $W$ of a Morse stack $F$ defined on a normal pseudomanifold $M$.}
    \label{alg:MorseWatershedAlgo}
\end{procedure*}

The soundness of this algorithm is a direct consequence of the above results. 
It may be implemented in linear time, with respect to the number 
of incidence relations of $M$, that is the cardinality of the set 
$\{(x,y) \; \mid \; x, y \in M$ and $x \subsetneq y \}$. 

\section{Morse watersheds and minimum spanning forests}
\label{sec:mst}

In \cite{cousty2009collapses}, an equivalence result which links
the notion of a watershed in an arbitrary stack with the one of a minimum
spanning forest is given. In this section, we refine this result in the case of Morse stacks.

Recall that we have defined a graph as a complex $X \in \bb{S}$ such that the dimension of $X$ is at most $1$. 

Let $X \in \bb{S}$ with $dim(X) = 0$, that is $X$ is a non-empty set of vertices. Let $Y$ be a graph such that $X \preceq Y$.
We say that $Y$ is a \emph{forest rooted by $X$}
if: 
\begin{itemize}
    \item we have $X = Y$, or
    \item there exists a free pair $(x,y)$ of $Y$ such that $Y \setminus \{x,y\}$ is a
forest rooted by~$X$. If $(x,y)$ is a free pair for $Y$, we say that $x$ is a \emph{leaf for $Y$}. 
\end{itemize}

If $X$ is made of a single vertex, then it may be seen that the previous definition is an inductive definition, which is equivalent
to the notion of a rooted tree in the sense of graph theory. If $X$ is made of $k$ vertices, then $Y$ has $k$ connected components.
Each of these connected components is a rooted tree for some vertex of~$X$. 


Let $M \in \bb{M}_d$.
The \emph{facet graph  of $M$} is the graph, denoted by $\Upsilon_M$, such that: 
\begin{itemize}
    \item A vertex $\{ x \}$ is in  $\Upsilon_M$ if and only if  $x$ is a $d$-face of $M$;
    \item An edge $\{x,y\}$ is in $\Upsilon_M$ if and only if $x \cap y$ is a  $(d-1)$-face of $M$.
\end{itemize}

Let $F$ be a Morse stack on $M \in \bb{M}$, and let $X = \{ \{ x \} \; \mid \;
x \in \mathfrak{min}(F) \}$. By Prop. \ref{pro:mw2},
each $\{ x \} \in X$ is a vertex of $\Upsilon_M$. Let $Y \preceq \Upsilon_M$ be a forest rooted by~$X$.
We say that $Y$ is a \emph{spanning forest for $\mathfrak{min}(F)$}
if all vertices of $\Upsilon_M$ are in~$Y$.  We define the \emph{weight of~$Y$} as the sum of all numbers $F(x \cap y)$, where $\{x,y\}$ is an edge of $Y$.
We say that $Y$ is a \emph{minimum spanning forest for $\mathfrak{min}(F)$}, if $Y$ is a spanning forest for $\mathfrak{min}(F)$
whose weight is minimum.

Let $F$ be a Morse stack on $M \in \bb{M}_d$.
We denote by $S$ the set of all couples of $d$-faces $(x,y)$ in $\bb{M}$ such that  $(x, x \cap y)$ is a differential pair of $F$,
 and $(x \cap y,y)$ is a flat pair of $F$.  Thus, $\pi = \langle x,x \cap y,y \rangle$
 is a $\nabla_d$-path  in $F$ from $x$ to $y$. 
\newline
The \emph{watershed forest of $F$} is the graph $G \preceq \Upsilon_M$ such that: 
\begin{itemize}
    \item All vertices of $\Upsilon_M$ are in $G$;
    \item An edge $\{ x,y \}$ is in $G$ if and only if $(x, y)$ or $(y,x)$ is a couple in $S$. 
\end{itemize}

From Proposition \ref{pro:mw3}, we can check that the watershed forest is indeed a forest. More precisely, we
can derive the following result.

\begin{Proposition} \label{pro:sf1}
If $F$ is a Morse stack on $M \in \bb{M}_d$, then
the watershed forest of $F$ is a spanning forest for $\mathfrak{min}(F)$.
\end{Proposition}

\begin{figure}
\begin{center}
 
    \includegraphics[width=1\columnwidth]{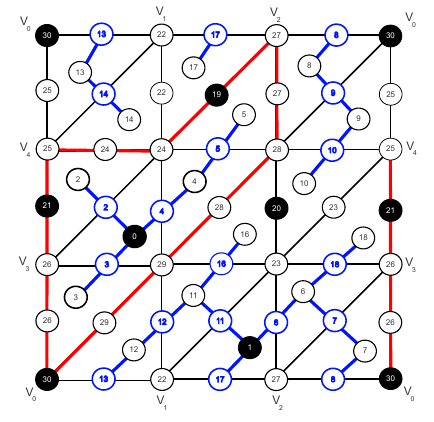} 
  \caption{\label{fig:MorseWatershedMSF}
    The Morse watershed (in red) on a Morse stack $F$ defined on a 2D-torus. In blue, the watershed forest, which is the minimum spanning forest for $\mathfrak{min}(F)$.
  }
\end{center}
\end{figure}

The following  optimality theorem~\ref{th:sf2} is illustrated in Fig.~\ref{fig:MorseWatershedMSF}.
\begin{theorem} \label{th:sf2} Let $F$ be a Morse stack on $M \in \bb{M}$.
The watershed forest of $F$ is a minimum spanning forest for $\mathfrak{min}(F)$.
Furthermore, the watershed forest of $F$ is the unique minimum spanning forest for $\mathfrak{min}(F)$.
\end{theorem}

\begin{proof}
Let $G$ be a minimum spanning forest for $\mathfrak{min}(F)$.
By a minimum spanning tree lemma \cite{cormen1999introduction}, \cite{motwani1995randomized},
if $\{ x \}$ is a vertex of $G$, then $G$ must contain
an edge $\{x,y\}$ that is a minimum weighted edge containing $\{x\}$. Now let $W$ be the watershed forest of $F$
and let $\{x,y\}$ be an edge of $W$. By the very definition of $W$, either $(x \cap y, x)$ or $(x \cap y, y)$ is a flat pair.
By the definition of a Morse stack, if $(x \cap y, x)$ is a flat pair, then $\{x,y\}$ is the only minimum weighted edge containing $\{x\}$.
Similarly, if  $(x \cap y, y)$ is a flat pair, then $\{x,y\}$ is the only minimum weighted edge containing $\{y\}$.
It follows that, if $e$ is an edge of $W$, then $e$ is the only  minimum weighted edge containing some vertex $v$ of $W$.
Since $v$ is necessarily in $G$, we deduce by the above lemma that $e \in G$. Therefore, we have $W \subseteq G$.
Since both $G$ and $W$ are spanning forests for $\mathfrak{min}(F)$, $G$ and $W$ have the same cardinality. Thus, we must have
$G = W$.
This shows that $W$ is the unique minimum spanning forest for $\mathfrak{min}(F)$. 
\end{proof}

\begin{figure*}[tb]
\begin{center}
    \includegraphics[width=.8\textwidth]{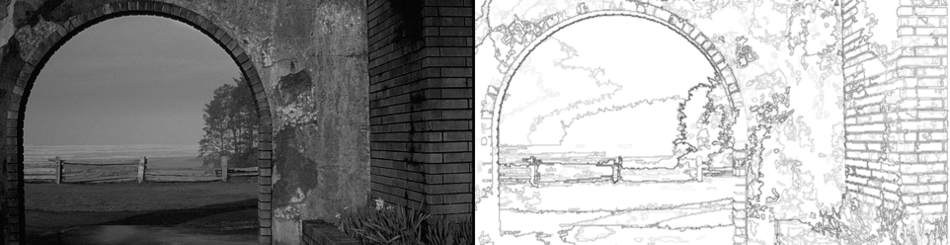} 
    
    \smallskip
    \includegraphics[width=.8\textwidth]{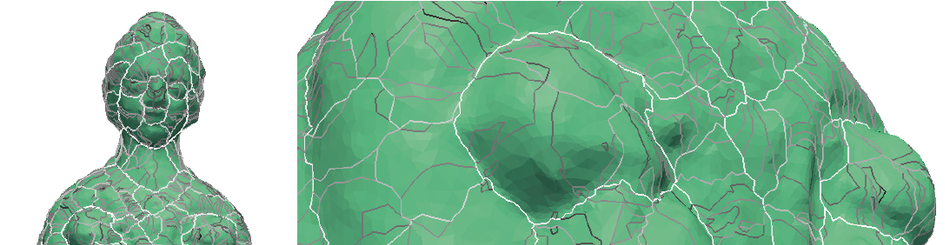} 
  \caption{\label{fig:saliency}
    Top: an image (left), together with a geodesic saliency map (right) where the darker a contour is, the more persistent it is. Bottom: two different views of a triangular mesh, superimposed with a geodesic saliency map where the whiter a contour is, the more persistent it is. 
  }
\end{center}
\end{figure*}
\section{Discussion, future work and conclusion}
\label{sec:discussion}

In this paper, we introduce Morse watersheds which satisfy a fundamental drop-of-water principle. As far as we know, this is the first definition of a watershed in the context of discrete Morse theory. These watersheds are based on Morse stacks, a class of functions that are equivalent to discrete Morse functions.
We show that the watershed of a Morse stack on a normal pseudomanifold is uniquely defined, 
and can be obtained with a linear-time algorithm relying on a sequence of collapses. 
Last, we prove that such a watershed is the cut of a unique minimum spanning forest, rooted in the minima of the Morse stack. 


While a watershed definition have been proposed in the continuous Morse setting \cite{najman1994watershed}, the watershed notion was mainly a source of inspiration in discrete Morse theory. We mention in particular the following.
\begin{itemize}
\item A watershed algorithm was used as a preprocessing in \cite{vcomic2016computing} for computing a gradient vector field.
Our approach directly provides watershed basins that are defined with gradient paths.
\item In \cite{delgado2014skeletonization}, watershed ideas were used as a motivation for obtaining Morse cells similar to catchment basins, with application to image segmentation. Our framework allows clarifying the difference between Morse cells and watershed basins. For example, in Fig.~\ref{fig:MorseWatershed}, the critical 1-simplex with altitude 20 is not part of the watershed cut, while it is part of the boundary of the Morse cell: it is a separating face, but it is not a biconnected one. We intend to explore in more details those differences in future work. We also envision studying the Morse-Smale decomposition.
\end{itemize}

One significant difference between Morse stacks and discrete Morse functions is that minima are $d$-dimensional simplices in our framework, while they are 0-dimensional ones in Morse theory. Although this might appear minor, such a difference has important consequences. In particular, the watershed is a pure $(d-1)$-subcomplex, while a similar property is not directly possible with the boundary of Morse cells, as classically  defined (see~\cite{delgado2014skeletonization} for example). Our approach allows for easily extracting topological features linking two regions, following the seminal paper \cite{najman1996geodesic}: indeed, we can for instance weight any simplex of the watershed cut with the persistence/dynamics \cite{boutry:hal-03676854} at which it disappears in a filtering. Such a representation, illustrated in Fig.~\ref{fig:saliency}, is called a {\em geodesic saliency map} in mathematical morphology, and is widely used (under the name {\em ultrametric contour map}~\cite{arbelaez2010contour}) as a post-processing step behind deep-learning approaches. See \cite{najman2011equivalence,cousty2018hierarchical} for theoretical studies of this notion, and \cite{perret2019higra} for a toolbox implementing many variations around it.

Data analysis heavily relies on data simplification and data visualization. We advocate that the watershed, together with filtering operators such as morphological dynamics, is a cornerstone for data analysis \cite{challa2019watersheds}. We aim at controlling the topological simplification, and understanding what is discarded in the simplification. The results of this paper are a first step in this direction. We envision using skeleton algorithms such as \cite{bertrand2014powerful}, and tools from cross-section topology \cite{bertrand1997image}, that, up to now, have been used mainly for image analysis. Indeed, these tools can be applied to general data. 
In this regard, an important perspective of the current paper is to bring together the topological data analysis framework with 
the mathematical morphology one.




\section*{Acknowledgements}
The authors would like to thank both Julien Tierny and Thierry Géraud, for many insightful discussions.

\begin{appendices}

\section{Normal pseudomanifolds} 
\label{app:pseudo}

A normal pseudomanifold is usually defined as a pseudomanifold that satisfies a certain link condition, which corresponds to a local property \cite{BAGCHI2008327}, \cite{BASAK2020107035}, \cite{DBNN2008}.
In this section, we show that this definition is equivalent to the one given in Definition~\ref{def:normal}.

Let $S$ be a finite set of simplexes.
If $x$ and $y$ are facets of $S$, a \emph{$p$-chain (in $S$) from $x$ to $y$} is a sequence
$\langle x = x_0,...,x_k = y \rangle$ of facets of $S$ such that, for each $i \in [0,k-1]$,
$x_i \cap x_{i+1}$ is a $q$-face of $S$, with $q \geq p$.
The set $S$ is \emph{$p$-connected} if, for any two facets $x,y$ in $S$, there is a $p$-chain in $S$ from
$x$ to $y$. 

We observe that: 
\begin{itemize}
    \item A complex is connected if and only if it is $0$-connected.
    \item A $d$-pure complex is strongly connected if and only if it is $(d-1)$-connected.
\end{itemize}


Let $X$ be a complex. Two faces $x,y \in X$ are \emph{adjacent} if $x \cup y \in X$.
The \emph{link of $x \in X$ in $X$} is the complex
$lk(x,X) = \{y \in X \; \mid \;  x \cap y  = \emptyset$ and $x \cup y \in X \}$. 
\newline
The \emph{star of $x \in X$ in $X$} is the set
$st(x,X) = \{y \in X \; \mid \;  x \subseteq y \}$.

Let $X$ be a $d$-pseudomanifold. We say that $X$ 
\emph{satisfies the link condition}  if $lk(x,X)$ is connected whenever
$x$ is a $p$-face of $X$ and $p \leq d-2$.

Let $X$ be a complex and $x$ be a $p$-face of $X$.
Let $st^*(x,X) = st(x,X) \setminus \{ x \}$.
We have
 $lk(x,X) = \{y \setminus x \; \mid \; y \in st^*(x,X) \}$
and $st^*(x,X) = \{z \cup x \; \mid \; z \in lk(x,X) \}$. 
\newline
We note that there is a set isomorphism between $lk(x,X)$ and $st^*(x,X)$, which preserves
set inclusion.
If $y \in st^*(x,X)$, the corresponding face $y \setminus x$ of $lk(x,X)$ is such that $dim(y \setminus x) = dim(y) - (p + 1)$.
Thus,  $dim(y \setminus x) = dim(y) - p^+$, where $p^+ = p +1$ is the number of elements in $x$.

Let $X$ be a $d$-pseudomanifold and $x$ be a $p$-face of $X$.
Let $p^+ = p+1$ and $d' = d - p^+$.
The following facts are a direct consequence of the above isomorphism:
\begin{itemize}
    \item The complex $lk(x,X)$ is $d'$-pure.
    \item The complex $lk(x,X)$ is non-branching.
    \item The set $st^*(x,X)$ is $q$-connected if and only if $lk(x,X)$ is $q'$-connected, with $q' = q - p^+$. 
\end{itemize}


\begin{Proposition} \label{pro:pseudonormal2}
A pseudomanifold is normal if and only if it satisfies the link condition.
\end{Proposition}

\begin{proof}
Let $X$ be a $d$-pseudomanifold. 
\begin{enumerate}
    \item Suppose $X$ satisfies the link condition and let $S$ be a connected open subset of $X$. 
\newline
Let $x$ and $y$ be two $d$-faces of $S$. By Remark \ref{rem:path}, there exists a
$p$-chain $\pi$ in $S$ from $x$ to $y$. Thus, $\pi = \langle x = x_0,...,x_k = y \rangle$ is a sequence of facets of $S$ such that, for each $i \in [0,k-1]$,
$x_i \cap x_{i+1}$ is a $q$-face of $S$, with $q \geq p$.
We choose $\pi$ such that $p$ is maximal and, if $p$ is maximal, such that the number $K(\pi)$ of $p$-faces $x_i \cap x_{i+1}$, with  $i \in [0,k-1]$,
is minimal.
If $p = d-1$, it means that $S$ is strongly connected; then we are done. 
\newline
Suppose $p < d-1$ and let $x_i, x_{i+1}$ such that $z = x_i \cap x_{i+1}$ is a $p$-face. 
\newline
Since $X$ satisfies the link condition, $lk(z,X)$ is connected.
By the isomorphism between $lk(z,X)$ and $st^*(z,X)$, it follows there is a $q$-chain
$\langle x_i = w_0,...,w_l = x_{i+1} \rangle$ in $st^*(z,X)$ with $q >p$.
Therefore,
$\pi' = \langle x = x_0,...,x_i = w_0,...,w_l = x_{i+1},...,x_k = y \rangle$ is a $p$-chain in $S$ from $x$ to $y$.
But we have $K(\pi') < K(\pi)$, a contradiction. Thus, each connected open subset of $X$
is strongly connected.

    \item Suppose $X$ is strictly connected. That is, any connected open subset of $X$ is $(d-1)$-connected.
Let $x$ be a $p$-face of $X$ with $p \leq d-2$.
The set $st(x,X)$ is a connected open subset of $X$, thus it is $(d-1)$-connected.
Since $p < d-1$, it means that $st^*(x,X)$ is $(d-1)$-connected. Therefore, $st^*(x,X)$ is strongly connected.
By the isomorphism between $lk(x,X)$ and $st^*(x,X)$, it follows that $lk(x,X)$ is strongly connected.
Thus $lk(x,X)$ is connected.
\end{enumerate}
\end{proof}

In the second part of the proof of Prop. \ref{pro:pseudonormal2}, we showed that $lk(x,X)$ is strongly connected.
Consequently, we have the following characterization  of a normal pseudomanifold.



\begin{Proposition} \label{pro:pseudonormal3}
 A pseudomanifold $X$ is normal if and only if, for each $p$-face $x$ of $X$, with $p \leq d-2$,
the complex $lk(x,X)$ is a pseudomanifold.
\end{Proposition}

\section{discrete Morse functions} 
\label{app:DMF}

Let us consider the following definition of a discrete Morse function:

\begin{definition}[Morse function]
Let $X$ be a complex and let $F$ be a map from $X$ to $\bb{Z}$.
We say that $F$ is a \emph{discrete Morse function on $X$} if any face of $X$ is in at most one covering pair $(x,y)$ in~$X$
such that $F(x) \geq F(y)$. If $F$ is a discrete Morse function, we say that such a pair is a \emph{regular pair of $F$}.
\end{definition}

It may be checked that this definition is equivalent to the classical one given by Forman (See Def. 2.1 and Lemma 2.5 of \cite{forman1998morse}).
\newline
In this way, the \emph{gradient vector field of a discrete Morse function $F$}, written
$\gradient(F)$, is the set composed of all regular pairs of $F$. 

The following restriction of a discrete Morse function will lead us to Morse stacks. 

We say that a discrete Morse function $F$ on $X$ is \emph{flat} if we have $F(x) = F(y)$ whenever $(x,y)$ is a regular pair of $F$,
that is, if each regular pair of $F$ is a flat pair of $F$. 

We can check that a map $F$ from $X$ to $\bb{Z}$ is a flat discrete Morse function if and only if: 
\begin{enumerate}
    \item Each  covering pair $(x, y)$ in~$X$ is such that $F(x) \leq F(y)$; 
    \item Each face of $X$ is in at most one flat pair of $F$. 
\end{enumerate}
Therefore, if we consider the function $-F$, we obtain the following:

\begin{Proposition}
Let $X$ be a complex and let $F$ be a map from $X$ to $\bb{Z}$.
The map $F$ is a Morse stack on $X$ if and only if the map $-F$ is a flat discrete Morse function on $X$.
\end{Proposition}

The following proposition claims that, up to an equivalence, we may
assume that any discrete Morse function is flat (see Def. 2.27 and Prop. 4.16 of \cite{scoville2019discrete}).

\begin{Proposition} [from {\cite{scoville2019discrete}}]
If $F$ is a discrete Morse function on $X$, then there exists a flat discrete Morse function  $G$ on $X$ such that,
for every covering pair $(x,y)$ in $X$,
we have $F(x) \geq F(y)$ if and only if $G(x) \geq G(y)$.
In other words, the function~$G$ is such that $\gradient(G) = \gradient(F)$.
\end{Proposition}
\end{appendices}


\bibliography{MorseWatersheds}

\end{document}